\documentclass[conference]{IEEEtran}
\usepackage{graphicx,cite}
\usepackage[cmex10]{amsmath}
\usepackage{amssymb}
\usepackage{array}
\usepackage{fixltx2e}
\usepackage{amsfonts}
\usepackage{lineno}
\usepackage{algorithm}
\usepackage[noend]{algpseudocode}
\usepackage{dsfont}
\usepackage{fancyhdr}
\usepackage{cancel}

\newtheorem{proposition}{Proposition}
\newtheorem{proof}{Proof}
\newtheorem{lemma}{Lemma}

\DeclareMathOperator{\e}{e}
\DeclareMathSymbol{-}{\mathalpha}{operators}{`\-}
\DeclareMathOperator*{\argmax}{\arg\!\max}

\begin{document}


\title{Strategic Availability and Cost Effective UAV-based Flying Access Networks: S-Modular Game Analysis}

\author{\IEEEauthorblockN{Sara Handouf and Essaid Sabir }
\IEEEauthorblockA{NEST Research Group, LRI Lab., ENSEM, Hassan II University of Casablanca, Morocco. \\
sara.handouf@ensem.ac.ma, e.sabir@ensem.ac.ma}
}

\maketitle

\begin{abstract}
Nowadays, ubiquitous network access has become a reality thanks to Unmanned Aerial Vehicles (UAVs) that have gained extreme popularity due to their flexible deployment and higher chance of Line-of-Sight (LoS) links to ground users. Telecommunication service providers deploy UAVs to provide flying network access in remote rural areas, disaster-affected areas or massive-attended events (sport venues, festivals, etc.) where full set-up to provide temporary wireless coverage would be very expensive. Of course, a UAV $i$s battery-powered which means limited energy budget for both mobility aspect and communication aspect. An efficient solution is to allow UAVs swhiching their radio modules to sleep mode in order to extend battery lifetime. This results in temporary unavailability of communication feature. Within such a situation, the ultimate deal for a UAV operator is to provide a cost effective service with acceptable availability. This would allow to meet some target Quality of Service (QoS) while having a good market share granting satisfactory benefits. In this article, we exhibit a new framework with many interesting insights on how to jointly define the availability and the access cost in UAV-empowered flying access networks to opportunistically cover a target geographical area. Yet, we construct a duopoly model to capture the adversarial behavior of service providers in terms of their pricing policies and their respective availability probabilities. Optimal periodic beaconing (small messages advertising existence of a UAV) is a vital issue that needs to be addressed, given the UAVs limited battery capacity and their recharging constraints. A full analysis of the game outcome, both in terms of equilibrium pricing and equilibrium availability, is derived. We show that the availability-pricing game exhibits some nice features as it is sub-modular with respect to the availability policy, whereas it is super-modular with respect to the service fee. Furthermore, we implement a learning scheme using best-response dynamics that allows operators to learn their joint pricing-availability strategies in a fast, accurate and distributed fashion. Extensive simulations show convergence of the  proposed scheme to the joint pricing-availability equilibrium and offer promising insights on how the game parameters should be chosen to efficiently control the duopoly game.\\

\end{abstract}
\section{Introduction}\label{sect:introduction}
 
In the recent few years, Unmanned Aerial Vehicles (UAVs), or drones,  have attracted a lot of attention, 
since they present ubiquitous usability, high maneuverability and low cost deployment advantages. UAVs, equipped with navigation systems and smart sensors, are currently being deployed for surveillance operations, rescue missions and rapid on demand communication. Along with the maturity of the UAVs  technology and relevant regulations,  they have become an important  market where a UAVs  worldwide deployment is expected. For instance, the registered number of UAVs in use in the U.S. exceed 200,000 just in the first 20 days of January 2016, after the Federal Aviation Administration (FAA) started requiring owners to sign up \cite{ding2017amateur}. Indeed, UAVs  can be  deployed  fruitfully to support cellular communication systems when  terrestrial infrastructure networks are damaged. Moreover it supports wireless communication in exceptional scenarios such as hard to reach rural areas, festivals or sporting events and emergency situations where the terrestrial base stations installation may be too expensive. In this context, UAVs as flying base stations have been reported as a promising approach that can boost the capacity and coverage performance of existing wireless networks. using UAVs as aerial base stations provides several benefits such as high maneuverability and robustness, flexible deployment, efficient on demand telecommunication and also mobility \cite{mozaffari2017wireless}.\\

Thanks to their high mobility and air location, UAVs are likely to have good communication channel since the UAV-ground link is more likely to have line-of-sight (LoS) links \cite{amorim2017radio}.  Furthermore, UAVs  provide unlimited telecommunication applications for their advantage to efficiently sense information about surrounding environment. They are now a key technology for intelligence, recognition, search, inspection tasks, surveillance,  public services,  and so on. Another key feature  of  UAVs is the Internet of Things (IoT) Which is  a technological revolution that has made the leap from conceptual to actual (see \cite{sharma2017saca} and \cite{mahmoud2015integrating}). IoT enables devices  to exchange data and inter-operate within the internet infrastructure, which affords ubiquitous connectivity while reducing the transmission cost and providing extended range for low- power communication. Technically, UAVs  play an important role in the Internet of Things vision.  It offers flexible deployability and re-programming in mission possibilities  to deliver numerous operating solutions and services, for IoT  partners.\\

While deploying drones as flying base stations offers numerous advantages, still a number of economic and technical challenges arise. These challenges comprise not only technical and physical issues, but also effective pricing and availability management features. Yet, for UAVs service providers to overcome these relevant challenging issues, availability or beaconing  scheduling becomes a very important but largely unexplored topic. For this, comprehensive modeling and performance analysis of UAVs setups has become extremely attractive. In this article we provide a UAVs joint pricing and availability problem, considering their limited battery capacity. We focus on the scheduling of availability periods as a key enabler that can potentially lead to better energy efficiency and then, to a better long-term availability (time-life) as well as a satisfactory quality of service. More precisely, we are interested in the joint Pricing-Availability problem for UAV-based network market. Then, we design a new scheme that jointly considers and solves the UAV's pricing and energy-efficient issue. \\

In order to address this issues, we examine the availability and pricing interactions between service providers as a non-cooperative game. We propose a duopoly setup, where a finite network of mobile unmanned aerial vehicles are deployed according to a homogeneous Poisson Point Processes (PPP), and serving a number of ground IoT devices. The UAVs are moving according to a Random Waypoint (RWP) model. As a first proposal, we derive the coverage probability expression as well as the service probability for each UAV $i$n the current scenario. In order to achieve the maximum system performance in terms of pricing policy and energy efficiency, our proposal introduces Nash equilibrium analysis. Furthermore, we propose a learning automata to derive the joint Price-Availability Equilibrium. Finally, we provide extensive numerical simulations to highlight the importance of taking Price - Availability as a joint decision parameters and provide thereby important insights/heuristics on how to set them.\\

The remainder of this article is organized as follows. In section II we give an overview of the proposed UAVs duopoly system model and strategic Pricing-Availability. In section III we present the Availability game analysis, existence and uniqueness of NE solution  and also define the game sub-modularity characteristics. Furthermore Pricing game is provided in section IV, the property of  super-modularity was discussed mathematically as well analytically. Section V considers a joint Availability-Pricing approach with numerical learning implementations, we further gave the impact on some important parameters on he learning process.  finally section VI concludes the paper.

\section{Related work} 

A significant existing literature investigates interesting features on UAVs technological performance and reducing cost. For instance, authors in \cite{zeng2017energy} study energy-efficient UAV communication  via optimizing the UAV's trajectory, for this, they consider a  UAV trajectory with general constraints, under which the UAV's flight radius and speed are jointly optimized. An important study of the UAV's obstacle avoidance plan  is conducted in \cite{li2017study}, where the writers propose algorithm based on Iterative Regional Inflation by Semidefinite in order to improve the fighting efficiency.  A  useful  derivative-free one, sliding mode control (SMC) theory-based learning algorithm,  for the control and guidance of a UAV $i$s proposed in \cite{kayacan2017learning}. Robustness, in real-time applications and finite time convergence  were proven trough intensive simulations. An other learning approach is treated in \cite{carrio2017review} where the authors provide a thorough review on deep learning methods for UAVs,  including the most relevant developments as well as their performances and limitation. An important issue which is IoT applications is discussed \cite{mozaffari2017mobile}. This paper investigate the efficient deployment and mobility of multiple unmanned aerial vehicles (UAVs), in order to enable reliable uplink communications for the IoT devices with a minimum total transmit power. Another IoT related work is in \cite{gharibi2016internet}, it presents the features that a UAV system should implement for drone traffic management, namely the air traffic control network, the cellular network, and the Internet. In \cite{kawamoto2014internet} the authors introduce different classifications of Internet of Things with examples of utilizing IoT technologies. They refer to the concept of global-scaled IoT, where using Unmanned Aerial Vehicle (UAV) offers a large coverage and more flexibility thanks to the mobility approach. The UAV based IoT concept has been further studied in paper \cite{motlagh2016low} that provides a comprehensive survey on the UAVs and the related issues. It highlights potential for the delivery of UAV-based value-added IoT service from height.

\section{Availability-Pricing Game Among UAV Operators} 

\subsection{The Flying Access Network Framework}
\subsubsection{Coverage probability}
Consider a circular geographical area with a radius R within which a number $N$ of wireless users are deployed according to a homogeneous Poisson Point Processes (PPP) with density of λu (number of users per m$^2$).
In this area, a set of UAVs moving randomly according to a Random Waypoint mobility model, are used as aerial mobile base stations to provide wireless service for the ground IoT devices (see figure 1). The drones are belonging to different operators and are engaged to provide an effective coverage for mobile IoT users. In this setting the drones have the same characteristics such as altitude h, total available bandwidth and maximum transmit power. 
\begin{figure}
\centering{
\includegraphics [width=9cm,  clip]{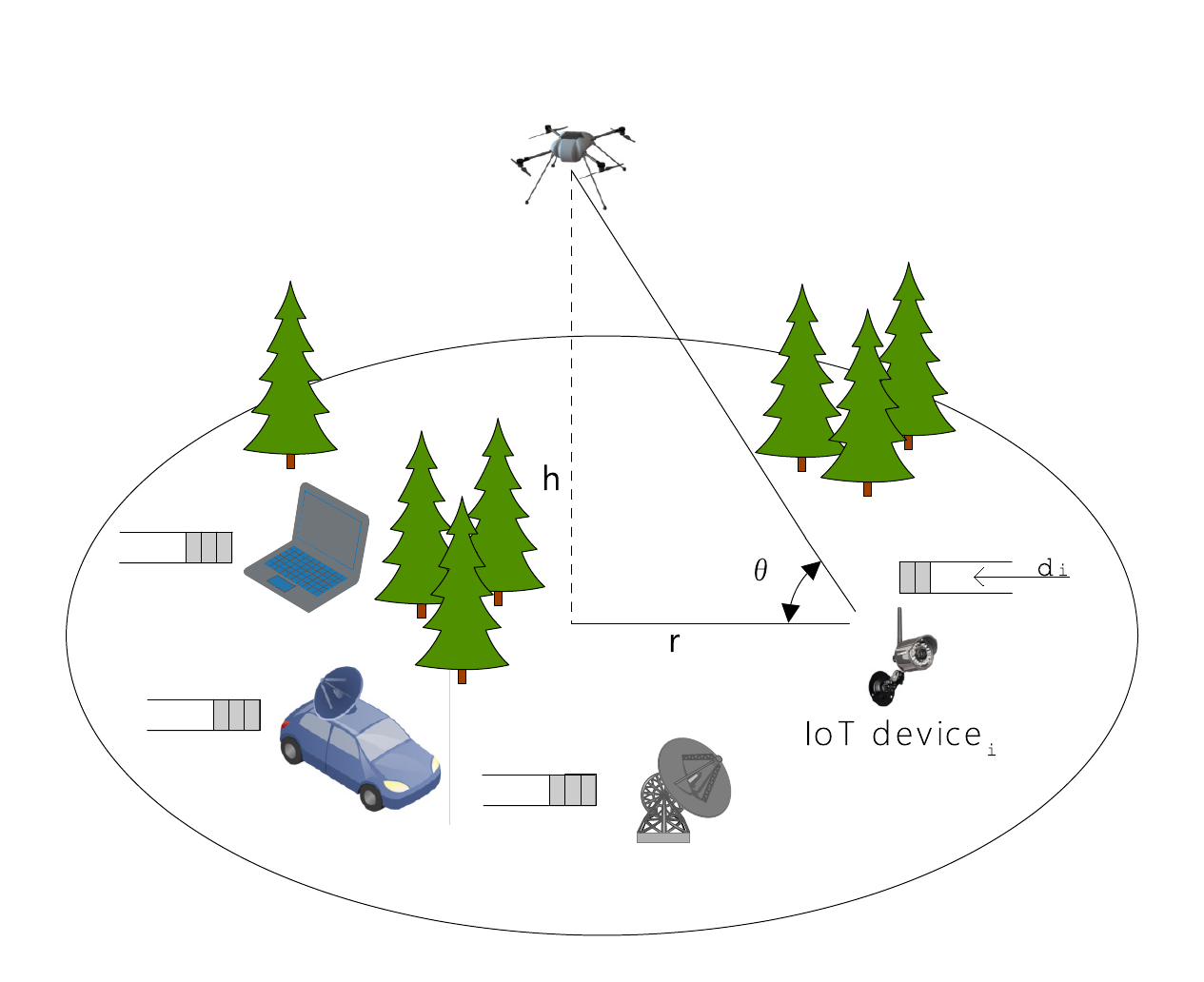}
\caption{Using UAV for Ubiquitous Network Access in IoT Environments.}}
\end{figure}

The coverage probability is formally defined as:
\begin{equation}
P_{u}^{cov}\left(\gamma_{u}\right)= \mathcal{P}\left(\gamma_{u} \geq \beta\right)
\end{equation}

Here $\gamma_{u}=\frac{P_{u}}{N+I}$ is the $SINR$ expression for a device user that connects to the UAV, where $P_{u}$ is the received signal power from the UAV $i$ncluding fading and path-loss, $N$ is the noise power and $I$  is the total interference power from existing transmitters. $\beta$ refers to the $SINR$ threshold.

According to \cite{mozaffari2016unmanned}, for an up-link user located at $(r, \phi)$, the coverage probability is given by:

\begin{equation}
\begin{split}
P_{u}^{cov}(r,\phi,\beta)=&P_{LoS}(r, \phi) \mathds{1} \left[ r \leq \left(\frac{P_{u}}{\beta N}\right)\right]+\\
&P_{NLoS}(r,\phi) \mathds{1} \left[ r \leq \left(\frac{P_{u}}{\beta N}\right)\right].
\end{split}
\end{equation}
Here $r$ and $\phi$ are the radius and angle in a polar coordinate system, the UAV $i$s located at the center of the area of interest.
$P_{LoS}$ respectively $P_{NLoS}$ are the line-of-sight, respectively Non-Line-of-Sight probabilities.

The average coverage probability is computed by taking the average of $P_{u}^{cov}(r,\phi,\beta)$ over the
area with the radius $R$:
\begin{equation}
\begin{split}
P_{u}^{cov}(r,\phi,\beta)= &\mathcal[E]_{r,\phi}[P_{u}^{cov}(r,\phi,\beta)]\\
&= \int_0^{min \left[\left( \frac{P_{u}}{\beta N}\right)^\frac{1}{\alpha_{u}},R\right]} \;\;P_{LoS}(r,\phi)\frac{2r}{R^2} \mbox{d}r\\
&+ \int_0^{min \left[\left( \frac{\eta P_{u}}{\beta N}\right)^\frac{1}{\alpha_{u}},R\right]} \;\;P_{NLoS}(r,\phi)\frac{2r}{R^2} \mbox{d}r,
\end{split}
\end{equation}

where  $\alpha_{u}$ is the path loss exponent over the user-UAV link. $\eta $  is an additional attenuation factor due to the NLoS connection.
$P_{u}$ is the UAV transmit power.\\

The air-to-ground signal propagation is almost affected  by obstacles in surrounding  environment. Thus   a   common  approach for channel modeling  between the UAV and down-link users is based on probabilistic direct Line of Sight (LoS) and Non Line of Sight (NLoS) links. Whereby each Link may occur with a specific probability which depends  on the propagation environment, the UAV and users locations, and the elevation angle (here $\theta$).

In this scenario we suppose designing a UAV-based communication system in a remote rural area where full set-up to provide temporary wireless coverage would be very expensive.  Due to its higher altitude and  the low average height of building in hard to reach rural areas, UAV has higher probability of direct  line-of-sight links to ground users. The probability of NLoS component occurrence  is therefore significantly lower than that of the LoS  component.  For simplicity  we ignore the small scale NLoS and as evident the coverage probability formula would be significantly simplified as bellow:
\begin{equation}
P_{u}^{cov}(r,\phi,\beta)=  \int_0^{min \left[\left( \frac{P_{u}}{\beta N}\right)^\frac{1}{\alpha_{u}},R\right]} \;\;P_{LoS}(r,\phi)\frac{2r}{R^2} \mbox{d}r.\\
\end{equation}

\subsubsection{Service probability}
Given the limited capacity and recharging difficulties of a UAV battery,  periodic beaconing is one of the most challenging problem that needs to be addressed. Each UAV $i$ is periodically sending beacons of duration $\tau_{i}$ advertising his presence to mobile IoT devices. The beacon/idle cycle is periodically repeated every time slot T during a time window: $m = l * T$ (see figure 2). UAVs  should define their beaconing periods strategically in order  to maximize their  encounter rate with the IoT ground users. However, they should avoid battery depletion resulting from maintaining useless beaconing in the absence of contact opportunities.
\begin{figure}
\centering{
\includegraphics [width=9cm,  height=4cm]{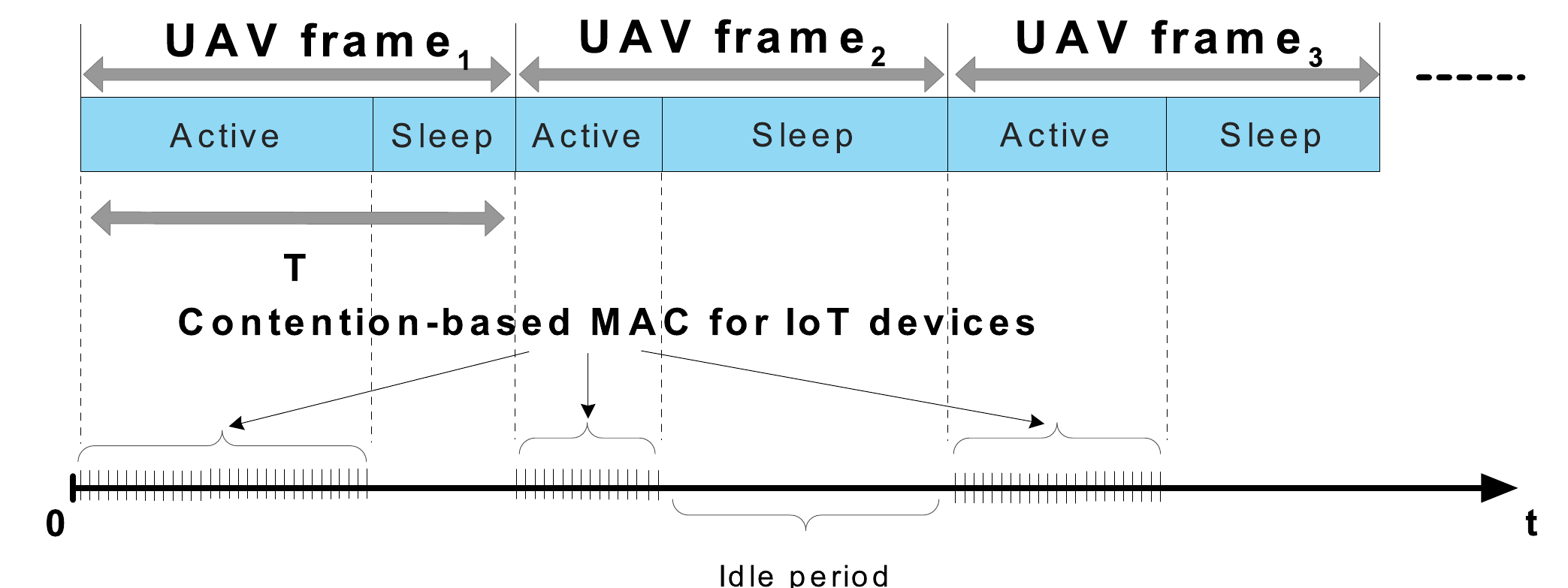}
\caption{beaconing cycle schedule }}\label{fig:zzz}
\end{figure}
The first encounter follows an exponential distribution with a random parameter $\lambda$. In order for a UAV $i$ to encounter first the ground IoT destination at time $\tau_{i}$, the following conditions must hold:
\begin{itemize}
\item The UAV $i$ has to be beaconing at $\tau_{i}$;
\item All potential encounters with ground destination, happening before time instant $\tau_{i}$, by the other UAVs need to be unsuccessful. In other words, the encounters need to happen while UAV $i$'s competitors are inactive).\\
\end{itemize}

Consequently, the successful contact probability  is given by:
\begin{equation}
P_{i}^{srv}=\left[P(T_{i} \leq T_{j})+P(T_{i} \geq T_{j}). P_{j}^{slp}\right].P_{i}^{bcn}.P_{i}^{cov}.
\end{equation}
For mathematical calculations, we focus on the case of two operators ($UAV_{i}$ and $UAV_{j}$), this minimizes mathematical complication but still allows us to analyze  the important features of Operators strategies.
We define the probability of $UAV_{i}$ beaconing while encountering for the first time the destination within [0; m]:
\begin{eqnarray}
P_{i}^{bcn} &=&\sum\limits_{s=0}^{l-1}\left(\int_{s*T}^{s*T+\tau_{i}} \lambda_{i} \e^{-\lambda_{i}x}dx\right)\nonumber\\
&=&-\frac{\e^{\lambda_{i}T}\left(\e^{-m\lambda_{i}}-\e^{-\lambda_{i}(m+\tau_{i})}-1+\e^{-\lambda_{i}\tau_{i}}\right)}{\e^{\lambda_{i}T}-1}.
\end{eqnarray}
For a $UAV_{i}$, the probability of being idle is given by:
\begin{eqnarray}
P_{i}^{slp}&=&\sum\limits_{s=0}^{l-1}\left(\int_{s*T+\tau_{i}}^{(s+1)*T} \lambda_{i} \e^{-\lambda_{i}x}dx\right) \nonumber\\
&=&\frac{\e^{\lambda_{i}T}\left(-\e^{-\lambda_{i}(m+\tau_{i})}+\e^{-\lambda_{i}(m+T)}+\e^{-\lambda_{i}\tau_{i}}- \e^{-\lambda_{i}T}\right)}{\e^{\lambda_{i}T}-1}.
\end{eqnarray}

The probability that $ UAV_{i}$ encounters first the ground destination without accounting for its
state (probing/idle) is expressed as follows:

\begin{equation}
 P(T_{i} \leq T_{j})=\frac{\lambda_{j}\e^{-m(\lambda_{i}+\lambda_{j})}+(-\lambda_{i}-\lambda_{j})\e^{-\lambda_{j}m}+\lambda_{i}}{\lambda_{i}+\lambda_{j}}
\end{equation}
And finally we define The probability that $ UAV_{i}$ encounters first the ground destination without accounting for its state :
\begin{equation}
P(T_{i} \geq T_{j})=\frac{\lambda_{i}\e^{-m(\lambda_{i}+\lambda_{j})}+(-\lambda_{i}-\lambda_{j})\e^{-\lambda_{i}m}+\lambda_{j}}{\lambda_{i}+\lambda_{j}}.
\end{equation}

\subsection{The Game Formulation}
\subsubsection{Customers behavior}
Most  telecommunications  markets are oligopolistic, this means, dominated by a certain number of service providers  called oligopolists. In order to maximize benefits, the competition between service providers is becoming increasingly tough.
From an oligopolist point of view, increasing  market share is the most important objective. Thus, service providers are supposed to define the optimal  pricing policy and the best QoS (represented by availability) requested  in order  to attract more and more customers. Because customers are  rational by nature, they are likely to register to an operator rather than the others or to stay at no subscription state according to their own satisfaction (servie price, Quality of Service, Quality of experience…).\\

In this paper, we assume a duopoly telecommunication market where just two Service Providers compete against one another to provide service for IoT ground  users. Each UAV $i$ belonging to SP $i$ picks its availability duration $\tau_i$ represented by the periodic beaconing time chosen within intervalle $[0, \;T]$,  and a service fee per data unit $f_{i}  \in [0,\;f_{max}]$. Formally, SPs are intelligent individuals in conflict, this situations is involving rational decisions. \\

Non cooperative Game theory provides the appropriate  tools for determining optimal behavior in such competitive environments. It provides mathematical methodology for modeling and analyzing interactions between intelligent rational players in interests conflict.
This  branch of game theory  deals largely with how intelligent individuals interact with one another in an effort to achieve their own goals. Players are selfish and each individual player makes decisions independently without forming alliances.\\

We capture the Pricing-Availability interaction among two UAV operators as a normal form game. Each SP seeks to choose the level of service allowing a reasonable battery lifetime while maximizing its monetary profits, given the action of the competitor. The non-cooperative game is formally defined as follows:
\begin{equation}
\mathcal{G}=\Big\langle\mathcal{N}, \; \mathcal{A},\; \mathcal{U}\Big\rangle.
\end{equation}
where:
\begin{itemize}
\item $\mathcal{N}=\{1, \; 2\}$ is the set of players (UAV operators);\vspace{0.2cm}
\item $\mathcal{A}=\Big\{[0,\;T]\times [0,\;f_{max}]\Big\}^2$;\vspace{0.2cm}
\item $\mathcal{U}=\left(U_1, U_2\right)$ is the vector of payoffs defined as the difference between their respective profits and costs.\\
\end{itemize}
Increasing the experienced market demand is the most important objective of each operator. Then, from an operator perspective, the question is to define the best beaconing duration and the best pricing policy that attracts a high portion of subscribers. Ideally, one UAV would like to simultaneously maximize the proposed fee  and   minimize the total consumed energy. However, these are two conflicting objectives, since ground users are likely to be rational. Inspired by Logit model \cite{jonard1999duopoly}, we model the probability $\pi_{i}$ that a given IoT device is registering  with service provider $i$ which can also be seen as the market share of a UAV $i$. The demand function follows an exponential distribution with parameter $\mu$. This last parameter is called the constant hazard or temperature, which defines the users rationality, so that, the  higher is $\mu$ , the less users are rational.
\begin{equation}
\pi_{i}=\frac{\exp\left(-\frac{f_{i}}{\mu}\right)}{\sum\limits_{j=1}^2 \exp\left(-\frac{f_{j}}{\mu}\right)}.
\end{equation}

The average number of customers that are registering to UAV $i$ is $N.\pi_{i}$. Because UAVs compete in quality of service and price, it seems to be obvious that the experienced demand  depends on these parameters. However, it  depends also on the environment variables.\\

\subsubsection{Service providers strategic decision}

We are interested in the problem of optimal Pricing-beaconing  in fixed-wing UAVs that should  avoid battery depletion resulting from maintaining useless activity duration in the absence of contact with the ground. 
 We propose a new design paradigm that considers both the  $UAV's$ Pricing and energy-efficient problems.\\

The UAV $i$s periodically sending beacons advertising his presence to mobile users on the ground.
From a UAV  perspective, there is a trade-offs between the strategic parameters $\tau$  and energy consumption. On one hand, as the beacon duration increases, the coverage chance $P_{cov}$ as well as the service probability increase. On the other hand, total energy consumption is mainly proportional to the activity period duration. The same for the service price, where increasing the pricing policy value may increase revenues  but in the same time negatively affect the experienced market demand. Here, the strategy space and the reward are common knowledge, however the chosen beaconing period and taxed fee are not since decisions are taken simultaneously.
The Utility of each UAV SP is a function of its own strategy as well as the decisions of the other. 
\begin{equation}
U_{i}(\tau_{i},\tau_{j},f_{i},f_{j})=\pi_{i}\sum\limits_{IoT_{i}\in \mathcal{R}_{i}}P_{i}^{srv}(\tau_{i},\tau_{j}).f_{i}-E_{i}.
\end{equation}

 $E_{i}$ is the dissipated energy for a $UAV_{i}$: 

\begin{equation}
\begin{split}
E_{i}=  &\left[ \left(1-\prod\limits_{i\in \mathcal{R}_{i}}\left(1-p_{u}\right)\right).\epsilon_{r}+\sum\limits_{u\in \mathcal{R}_{i}}\theta_{u,i}\epsilon_{Ack}\right].P_{i}^{srv}+\\
&\epsilon_{b}.\pi_{i}(\tau)+\epsilon_{s}.\\
\end{split}
\end{equation}
Here $\theta_{u,i}= p_{u} \prod\limits_{k \neq u}(1-p_{k})$ is the normalized throughput  of user $u$ served by a UAV $i$. We denote by $\epsilon_{b}$, the energy cost per slot for sending beacons, $\epsilon_{r}$ is the reception energy, $\epsilon_{Ack}$ is the transmission energy cost and $\epsilon_{s}$ is the unit energy for remaining swhiching the transceiver status. $\mathcal{R}_{i}=(User_{u}, d(User_{u},UAV_{i})\leq r)$ is the set of ground IoT devices served by a given drone $i$.
 \section{ Availability game with fixed fee }
In literature the pricing shelduding is almost the key parameter that affects straightly operators revenues. However the realism of this assumption is sometimes questioned \cite{baslam2011joint}. Indeed, similar market price may affect the customers perceptions about the UAVs' availability and then, their loyalty. Due to a rough competition, availability has started to become an important strategic tool for operators to expand their market share. We examine energy-efficient-Pricing trade-offs  within  a non-cooperative game framework. This incites operators  to make efforts in providing better QoS while fixing optimal periodic  advertising beacons. The proposed approach is  realistic enough to analyze the interplay of Availability on service providers revenues.
The problem is then, to determine for each UAV, the optimal availability (beaconing period) while rival behaves optimally, this represents exactly the Nash equilibrium status.\\

Recall  that  Nash equilibrium  strategy is the best response of each UAV to the expected best-response  behavior of the other.\\

\begin{proposition} The availability game admits at least one Nash Equilibrium.\\
\end{proposition}
\begin{proof}
To show the existence of an equilibrium, the sufficient condition is the quasi-concavity of the utility function. We notice that, the strategy set  $[0;T]$ is a convex, close and compact interval. Furthermore, the utility function is continue with respect to $\tau_{i}$.
In addition we show that the second derivative with $\tau_{i}$ is negative. After calculation, the second derivative is expressed as bellow:
\begin{equation}
\frac{\partial^2 U_{i}(\tau_{i},\tau_{j})}{\partial \tau_{i}^2}=\frac{\partial^2 P_{i}^{srv}(\tau_{i},\tau_{j})}{\partial^2 \tau_{i}}\left(N f_{i}\pi_{i}-\psi_{i}\right),
\end{equation}

where
\begin{equation}
\psi_{i}=\left[ \left(1-\prod\limits_{i\in \mathcal{R}_{i}}(1-p_{u})\right).\epsilon_{r}+\sum\limits_{u\in \mathcal{R}_{i}}\theta_{u,i}\epsilon_{Ack}\right].
\end{equation} 
 
It has been easily proven  that the second derivative of the service probability  with respect to $\tau_{i}$ is of negative sign, and  therefore  for the second derivative to be negative, it could be sufficient to make sure that the following condition fulfill
\begin{equation}
N f_{i}\pi_{i}-\psi_{i} \geq 0.
\end{equation}

This condition seems logical since the demand is positive and each operator requires a certain demand threshold  as an incentive to deploy a drone resource, this means that for a UAV operator to be motivated to coverage a certain area the
minimum required demand  equals  at least  the cost of energy that he is going to waste. 
From this, the utility function of a UAV $i$ belonging to an operator $i$ is quasi concave, and then admits a local optimum point at $\tau_{i}^*$ which is its best response against strategies of other  operators.\\
\end{proof}

\begin{proposition}[Uniqueness of Availability-termed Nash equilibrium] The availability game is dominance solvable \cite{elkington2009strict} and consequently, satisfies Rosen's diagonal strict concavity condition. Henceforth it has a unique Nash equilibrium.\\
\end{proposition}
\begin{proof}
The dominance solvability condition is given by:
\begin{equation}
-\frac{\partial^2 U_{i}(\tau_{i},\tau_{j})}{\partial \tau_{i}^2}-\left|\frac{\partial^2 U_{i}(\tau_{i},\tau_{j})}{\partial \tau_{j} \partial \tau_{i}}\right| \geq 0,
\end{equation}
when
\begin{equation}
\frac{\partial^2 U_{i}(\tau_{i},\tau_{j})}{\partial \tau_{i}^2}=\frac{\partial^2 P_{i}^{srv}(\tau_{i},\tau_{j})}{\partial^2 \tau_{i}}(N f_{i}\pi_{i}-\psi_{i}),
\end{equation}
and
\begin{equation}
\frac{\partial^2 U_{i}(\tau_{i},\tau_{j})}{\partial \tau_{j} \partial \tau_{i}}=\frac{\partial^2 P_{i}^{srv}(\tau_{i},\tau_{j})}{\partial \tau_{j} \partial \tau_{i}}(N f_{i}\pi_{i}-\psi_{i}).
\end{equation}
The  dominance solvability conditions are verified and consequently, the concave game $\mathcal{G}$  satisfies Rosen's diagonal strict concavity condition for the uniqueness of the Nash equilibrium.\\
\end{proof}

Interactions among UAVs in this  non-cooperative game have a very attractive property, if  one UAV reduces its availability (beaconing period), the other UAV also has an interest in decreasing its own. Informally, this type of games led us to conclude that it is about sub-modular games.\\

\begin{lemma} The game $\mathcal{G}$ with $N$ mobile devices and two UAVs service providers is sub-modular.\\
\end{lemma}
\begin{proof} To prove the Sub-modularity of the game, we compute the second order mixed derivative of $U_{i}(\cdot)$ with respect to two beaconing durations $\tau_{i}$ and $\tau_{j}$. With the mixed derivative expression given in (19), it is straightforward to conclude
that it is effectively about a sub-modular game.\\
\end{proof}

With sub-modularity, the best response of a UAV $i$s a decreasing  function of the other UAV availability strategy (beaconing period). Consequently, this game possesses a unique equilibria that best response dynamic learning  scheme can reach with probability 1.

\section{ Pricing game with fixed availability period}

Despite the several advantages of availability competition, still the price is a more common strategical decision in telecommunication business. This motivates us to investigate  the pricing game with fixed availability parameter. We argue in this section that, operators set prices and let the consumers decide on their experienced demand. That is, each operator is maximizing its profit believing that its competitor is doing the same. Conceptually, this situation is a typical Bertrand-Nash equilibrium (see papers in \cite{vives1990nash} and \cite{vives1984duopoly}). We turn first to mathematically analyze the equilibrium existence.\\

\begin{proposition} The Pricing game admits at least one Nash Equilibrium.\\
\end{proposition}
\begin{proof} The strategy set represented in $[0;f_{max}]$ is convex, closed and compact interval.
In addition the utility function is continue with $f_{i}$. Furthermore, we show that the second derivative with respect to $f_{i}$ has a negative sign.
After calculation, 
\begin{equation}
\frac{\partial \pi_{i}(f_{i},f_{j})}{\partial f_{i}}=-\frac{\frac{1}{\mu}  \exp\left(-\frac{f_{i}}{\mu}\right)\exp\left(-\frac{f_{j}}{\mu}\right)}{\left[\exp\left(-\frac{f_{i}}{\mu}\right)+\exp\left(-\frac{f_{j}}{\mu}\right)\right]^2},
\end{equation}
\begin{equation}
\frac{\partial^2 \pi_{i}(f_{i},f_{j})}{\partial f_{i}^2}=\frac{\exp\left(-\frac{f_{i}}{\mu}\right)\exp\left(-\frac{f_{j}}{\mu}\right)\left[\exp\left(-\frac{2f_{j}}{\mu}\right)-\exp\left(-\frac{2f_{i}}{\mu}\right)\right]}{\mu^2\left[\exp\left(-\frac{f_{i}}{\mu}\right)+\exp\left(-\frac{f_{j}}{\mu}\right)\right]^4}.
\end{equation}
The second order derivative is expressed as bellow:
\begin{eqnarray}
\hspace{-0.5cm}\frac{\partial^2 U_{i}(f_{i},f_{j})}{\partial f_{i}^2}&=& N.P_{i}^{srv}(\tau_{i},\tau_{j})\cdot\left[2\frac{\partial\pi_{i}}{\partial f_{i}}
+ f_{i} \frac{\partial^2 \pi_{i}}{\partial^2 f_{i}}\right]\nonumber\\
&=&\frac{N.P_{i}^{srv}\exp\left(-\frac{f_{i}}{\mu}\right)\exp\left(-\frac{f_{j}}{\mu}\right)}{\mu\left(\exp\left(-\frac{f_{i}}{\mu}\right)+\exp\left(-\frac{f_{j}}{\mu}\right)\right)^2}.\nonumber\\
&&\left[\frac{f_{i}\left(\exp\left(-2\frac{f_{j}}{\mu}\right)-\exp\left(-2\frac{f_{i}}{\mu}\right)\right)}{\mu\left(\exp\left(-\frac{f_{i}}{\mu}\right)+\exp\left(-\frac{f_{j}}{\mu}\right)\right)^2} - 2\right],
\end{eqnarray}
which is clearly negative. Hence the profit function of each UAV $i$s quasi-concave with respect to its own fee, and then the NE existence. \\
\end{proof}
\begin{proposition}[Uniqueness of Price-based Nash equilibrium] The Pricing game has a unique Nash Equilibrium.\\
\end{proposition}
\begin{proof}
The utility function satisfies the dominance solvability conditions and consequently satisfies  Rosen’s conditions defined as: 
\begin{equation}
\begin{split}
-\frac{\partial^2 U_{i}(f_{i},f_{j})}{\partial f_{i}^2}-\left|\frac{\partial^2 U_{i}(f_{i},f_{j})}{\partial f_{j} \partial f_{i}}\right|\\
=\frac{N.P_{i}^{srv}\exp\left(-\frac{f_{i}}{\mu}\right)\exp\left(-\frac{f_{j}}{\mu}\right)}{\left[\exp\left(-\frac{f_{i}}{\mu}\right)+\exp\left(-\frac{f_{j}}{\mu}\right)\right]^2} \geq 0.
\end{split}
\end{equation}
\end{proof}

A remarkable property of this game is that when one player picks a higher action, the others has incentive to follow and do the same.  Roughly, this introduces the notion of a super-modular game characterized by strategic complementarities.\\

\begin{lemma} The game $\mathcal{G}$ with $N$ mobile devices and two UAVs service providers is super-modular.\\
\end{lemma}
\begin{proof}
 To prove the Super-modularity of the game, we compute the second order mixed derivative of $U_{i}$ with respect to  its proposed fee ant that of its competitor ($f_{i}$ and $f_{j}$).
With the following mixed derivative expression, which is clearly of positive sign, we argue that it is a super-modular game.
\begin{equation}
\begin{split}
\frac{\partial^2 U_{i}(f_{i},f_{j})}{\partial f_{j} \partial f_{i}}=&\frac{N.P_{i}^{srv}\exp\left(-\frac{f_{i}}{\mu}\right)\exp\left(-\frac{f_{j}}{\mu}\right)}{\mu\left[\exp\left(-\frac{f_{i}}{\mu}\right)+\exp\left(-\frac{f_{j}}{\mu}\right)\right]^2}.\\
&\left(1-\frac{f_{i}\left[\exp\left(-2\frac{f_{j}}{\mu}\right)+\exp\left(-2\frac{f_{i}}{\mu}\right)\right]}{\mu\left[\exp\left(-\frac{f_{i}}{\mu}\right)+\exp\left(-\frac{f_{j}}{\mu}\right)\right]^2}\right).
\end{split}
\end{equation}
\end{proof}
Super-modularity tends to be analytically appealing. An interesting property is that it behaves well under various learning rules. Furthermore, convergence towards Nash equilibrium point is guaranteed under best response dynamics.

\section{Joint Pricing-Availability game}
\subsection{Insights on Real-World Implementation: Fully Distributed Learning}

Learning is a fundamental component of intelligence and cognition; it is defined as the ability of synthesizing the
acquired knowledge through practice, training, or experience in order to improve the future behavior of the learning agent.
In this section we introduce a learning scheme \cite{handouf2016telecommunication}, \cite{handouf2016pricing} aiming to understand the behavior of users during the interactions and to reach  the eventual convergence toward Nash Equilibrium. The goal is that UAVs service providers learn their own payoffs and determine the NE  joint Availability-pricing strategy.
To accomplish this, Best Response Algorithm is a suitable class; it leads to Nash Equilibrium status by exploiting fast monotony of best response functions. In fact,  NE concept is implicitly based on the assumption that players follow best-response dynamics until they reach a state from which, no player can improve his utility by changing unilaterally his strategy \cite{nash1950equilibrium}.\\

Algorithm 1 summarizes the best response dynamics that each  service provider performs in order to converge to the game Nash equilibrium. The best-response strategy of a player is defined as its optimal one, given the strategies of the other players. At a time iteration $t$, each service provider chooses the best joint Availability-Pricing strategy against the opponents strategies chosen at previous round  $t_{-1}$. It is worth to notice that convergence of the best response dynamics is granted for s-modular games \cite{Topkis98}. 
\begin{algorithm}
\caption{Best Response Dynamics}
\begin{algorithmic}
\State\textbf{Initialisation:}
\State {Service Provider $i$ initializes $f_{i}$ and $\tau_{i}$} at random values;
\State\textbf{Learning pathern:}
\State{for t=0,1,2,....m}\\
\State 
\begin{equation}\left(f_{i}^{t+1},\tau_{i}^{t+1}\right)= \argmax\limits_{f_{i}, \tau_{i}} \,U_{i}^{t}\left(f_{i},f_{j}^{t},\tau_{i},\tau_{i}^{t}\right).
\end{equation}
\State{end for}
\end{algorithmic}
\end{algorithm}

Best Response Algorithm exhibits enormous advantages as it  offers an accurate and fast convergence to the unique Nash equilibrium point. We say UAVs learn to play an equilibrium, if after a given number of iterations, the strategy profile converges to an equilibrium status. 

\subsection{Numerical Investigations}
The Pricing-Availability Best-response learning results are provided using MATLAB software. 
Parameters used in the simulation are set based on typical values such as $T = 1, m = 100, f_{min}=0$ and $f_{max} = 10$.
Here, we will analyze the impact of the various parameters especially market temperature, coverage probability and encounter rates  on the Nash equilibrium Price  as well as the provided QoS (measured by beaconing periods) for both competitors.\\

Under this setup, Figure 3 illustrates the behavior of service
providers under best response algorithm and the convergence
towards the Nash Equilibrium Price with time iterations. Here we
consider symmetric operators with randomly generated initial
points. A smaller temperature $\mu$ value means that consumers
are more and more rational so that the operator’s NE Price
tends towards lower values. Indeed with $\mu = 2$ the NE price is
absolutely zero $(=f_{min})$. However when $\mu$ becomes greater (consumers
are less rational) the NE price arises and reaches it’s maximal value with $\mu = 6$.
\begin{figure}
\centering{
\includegraphics [width=9.5cm,  height=5cm]{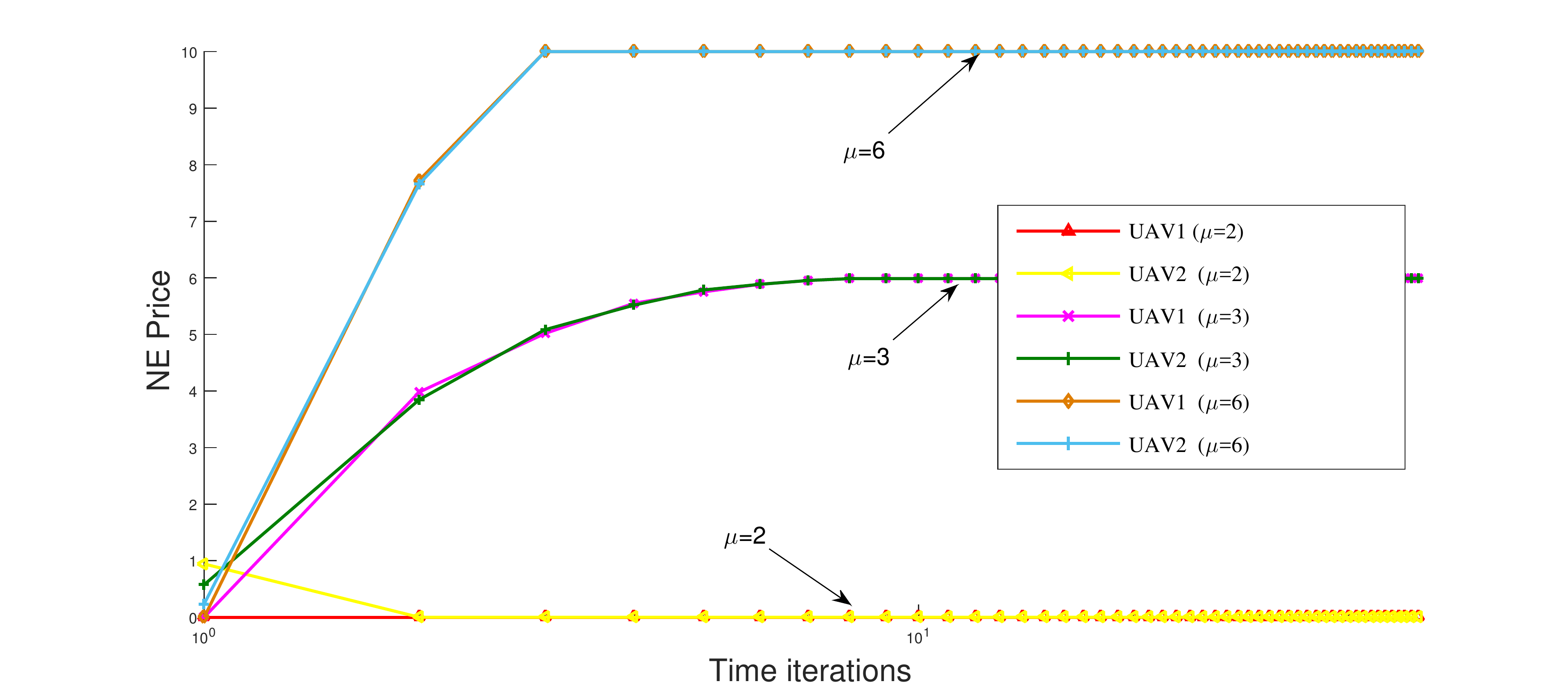}
\caption{NE price with temperature parameter under best response dynamic for two homogeneous operators}}\label{fig:NET}
\end{figure}
The same applies to Figure 4, where we depict the convergence of Algorithm 1 towards  Nash Equilibrium beaconing periods with different values of $\mu$. After some time iterations the two UAVs converge accurately to
the same Nash equilibrium beaconing period. With $\mu = 2$ the $\tau^{*}$ is zero however $\tau^{*}$ is increasing with $\mu$. By linking that with the NE price analysis, we can say that, with an increasing Availability periods, operators propose greater prices.\\
\begin{figure}
\centering{
\includegraphics [width=9.5cm,  height=5cm]{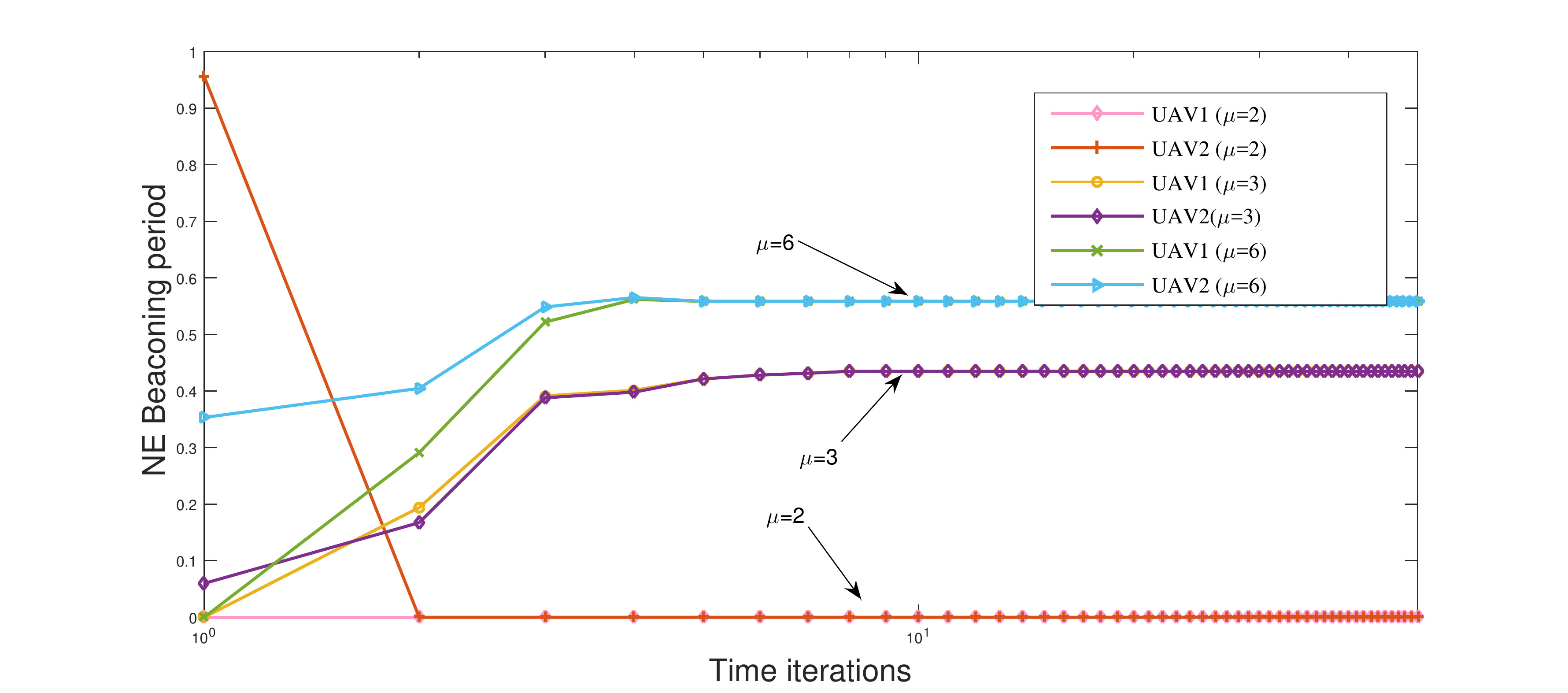}
\caption{NE beaconing period with temperature parameter under best response dynamic for two homogeneous operators}}
\end{figure}

We depict in figures 5 and 6 the individual best response prices respectively beaconing periods for two asymmetric UAVs. Each UAV has a different coverage probability value ( may be due to high parameter for example).  
Counterintuitively, figure 5 shows that the NE price does not depend on the coverage probability. Even with different coverage probabilities, the two competitors converge towards the same equilibrium price, this last value is increasing with UAV's coverage which is quite intuitive.\\

\begin{figure}
\centering{
\includegraphics [width=9.5cm,  height=5cm]{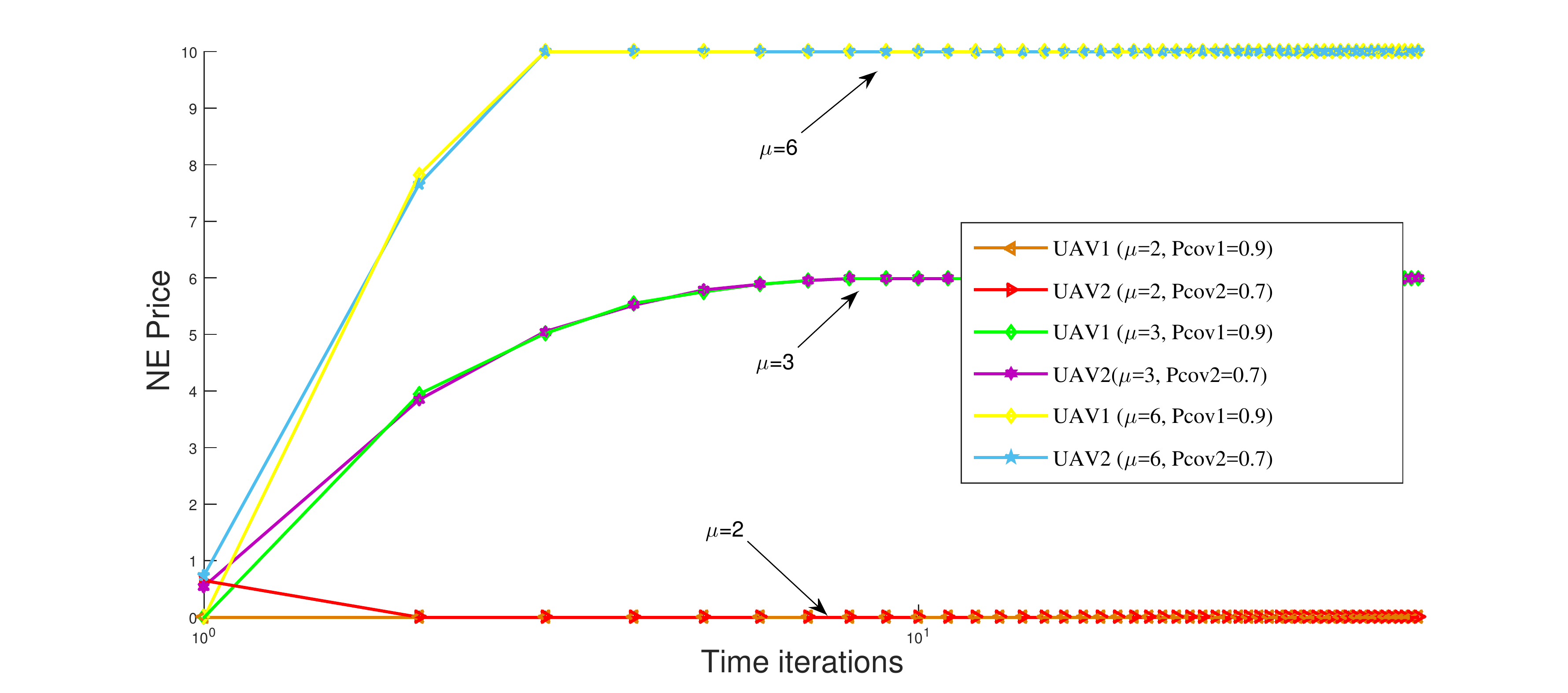}
\caption{NE beaconing period under best response dynamic for  inhomogeneous operators coverages }}
\end{figure}
However, equilibrium beaconing periods illustrated in figure 6 are different with UAVs of different coverage probabilities. It is worth notice that, the UAV with greater coverage, converges to a higher beaconing NE value, thus, offers a better QoS (availability service). We we exclude from this analysis  the case with small value of parameter $\mu$ where the two operators converge to null values. Yet the greater is the  coverage (likely UAV high), the longer is the NE beaconing availability advertising periods.\\

\begin{figure}
\centering{
\includegraphics [width=9.5cm,  height=5cm]{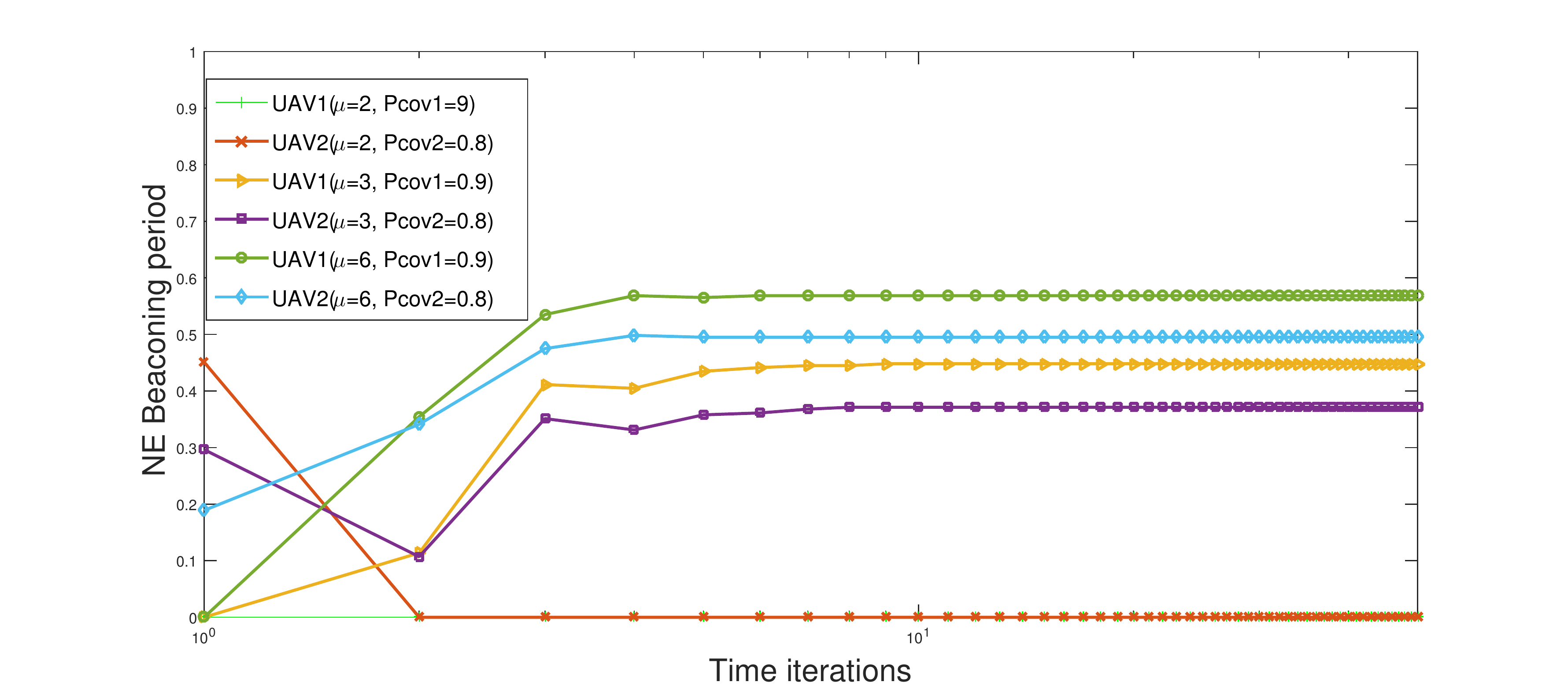}
\caption{NE price under best response dynamic for  asymetric  operators coverages}}
\end{figure}

Figure 7 illustrates the impact of both $\mu$ and the encounter rate $\lambda_{i}$ on the NE price on each UAV $i$ in an asymmetric setup (One UAV have an encounter rate advantage over the other).
We remark that, even with different encounter rates, the NE prices are the same however the NE with a minimum temperature is greater than that of symmetric case.\\
\begin{figure}
\centering{
\includegraphics [width=9.5cm,  height=5cm]{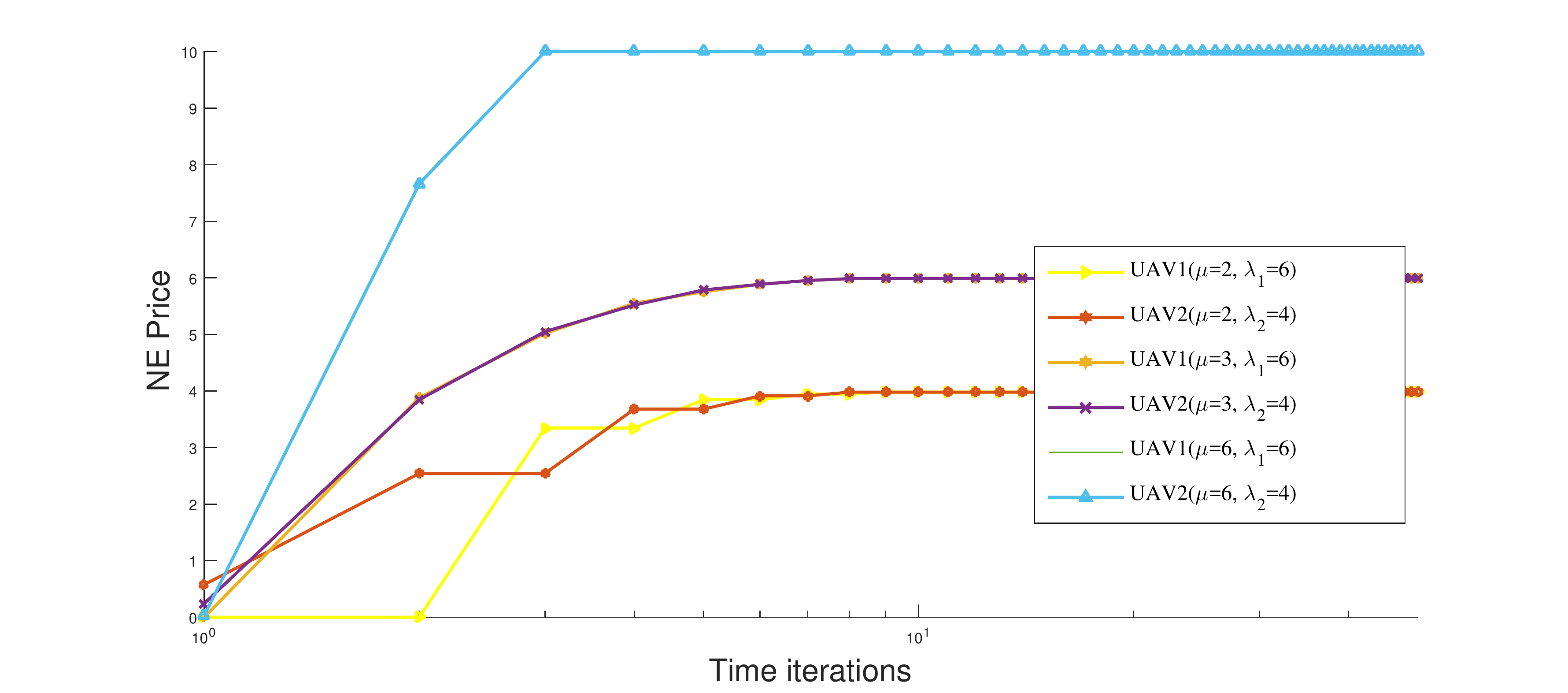}
\caption{NE prices under best response dynamic for inhomogeneous operators encounter rates }}
\end{figure}

We plot in figure 8 the beaconing best response learning evolution for two inhomogeneous UAVs (One UAV have an encounter rate advantage over the other), the NE beaconing periods are straightforward affected by $\mu$ as well as $\lambda$. Thus, with a greater $\mu$ and $\lambda$ the best response converge to a high value of $\tau^*$ however  the NE beaconing period for a UAV with smaller $\lambda$ is always lower.\\

\begin{figure}
\centering{
\includegraphics [width=9.5cm,  height=5cm]{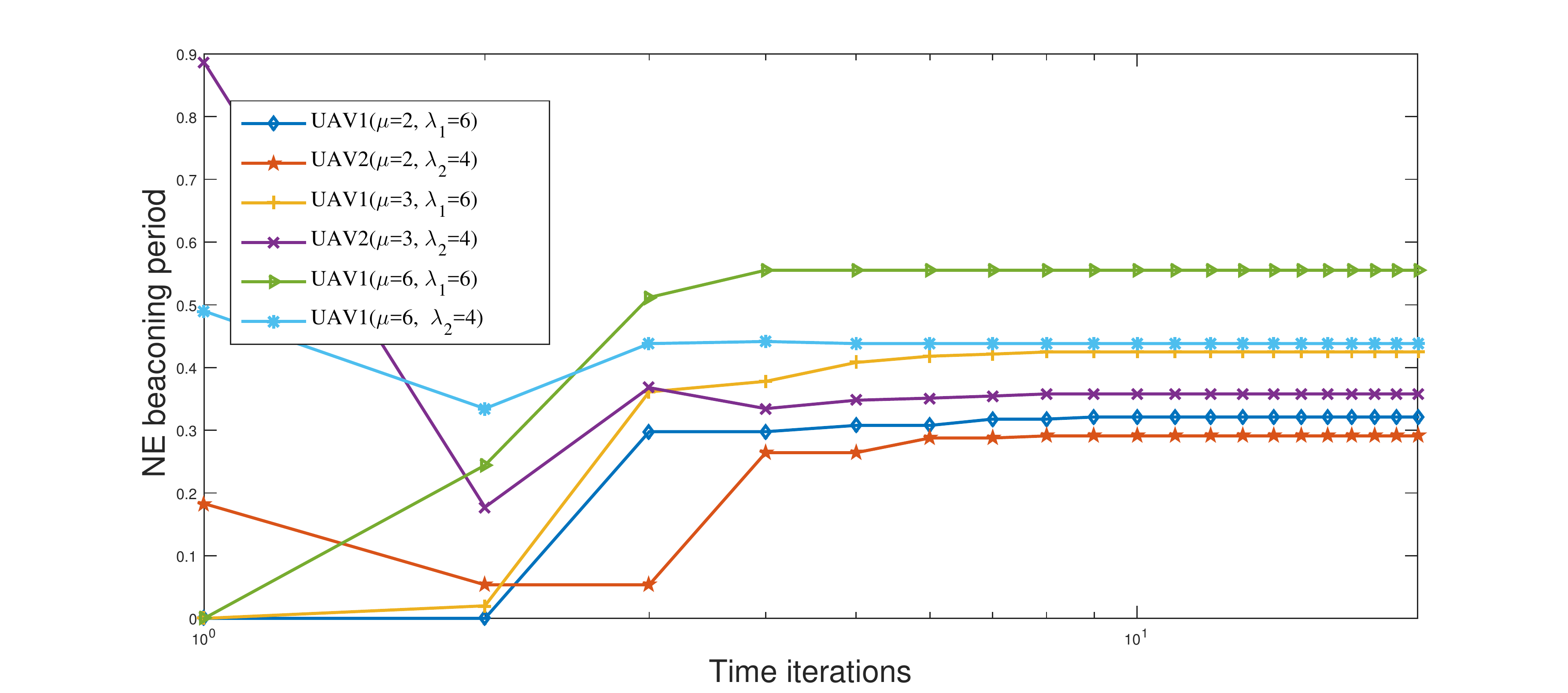}
\caption{NE beaconing periods under best response dynamic for  inhomogeneous operators encounter rates }}
\end{figure}

In order to study the  evolution impact of parameter $\mu$ on the learning results, figure 9 respectively figure 10 illustrates the NE price respectively NE beaconing period versus temperature parameter for two symmetric UAVs and different IoT population sizes.
The results in figure 9 show that with temperature parameter smaller than a certain threshold(for example $\mu = 2$ for $N=50$) the NE price is in its minimum value, and from this point the price is increasing until  the maximum price value ($f_{max}$) (for example $\mu$ = 6 for N=50). In this figure we argue that, with a greater population density the price tends to converge faster to the maximum value. That findings  impact obviously the beaconing period evolution illustrated in figure 10. With a greater population the NE Beaconing increases and reaches the equilibrium status in a smaller time window.\\
\begin{figure}
\centering{
\includegraphics [width=9.5cm,  height=5cm]{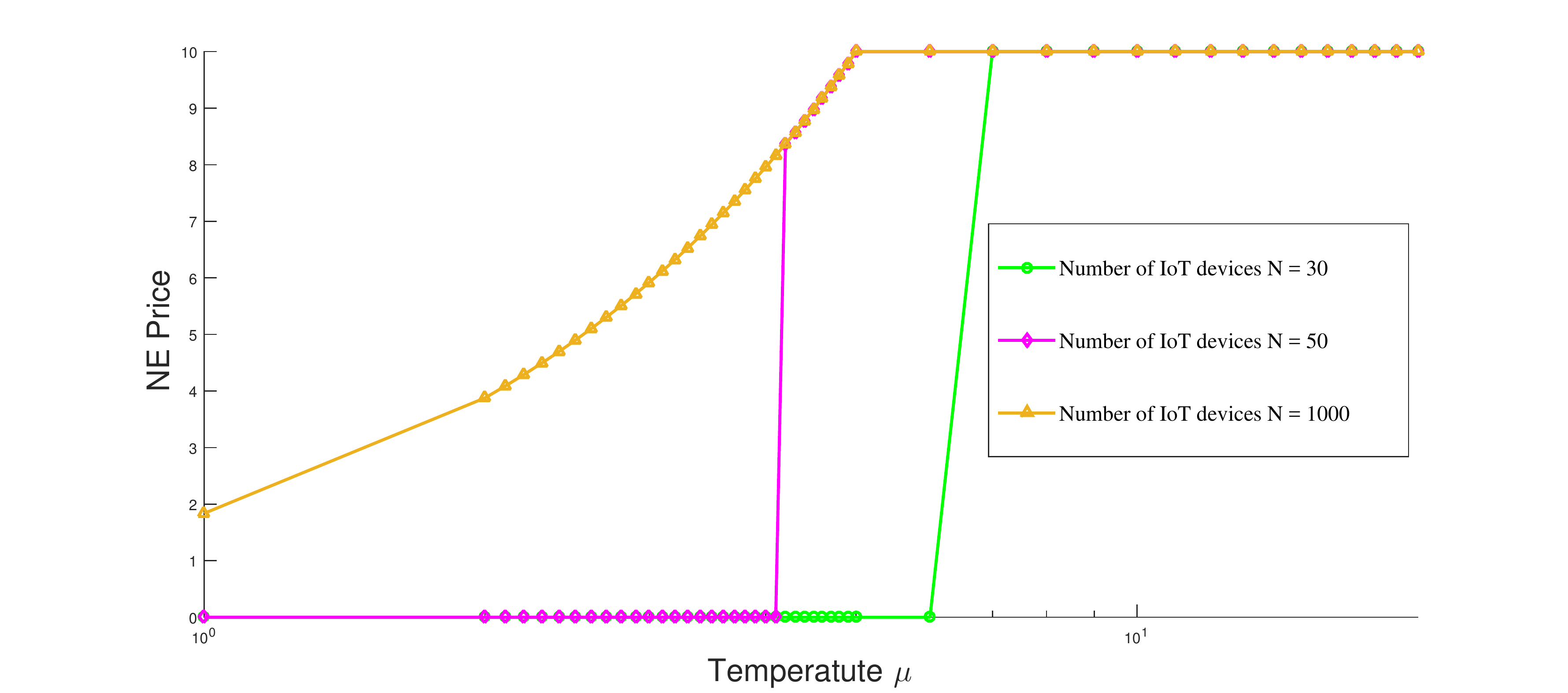}
\caption{NE Price evolution  for operators with respect to $\mu$ in different populations sizes}}
\end{figure}
\begin{figure}
\centering{
\includegraphics [width=9.5cm,  height=5cm]{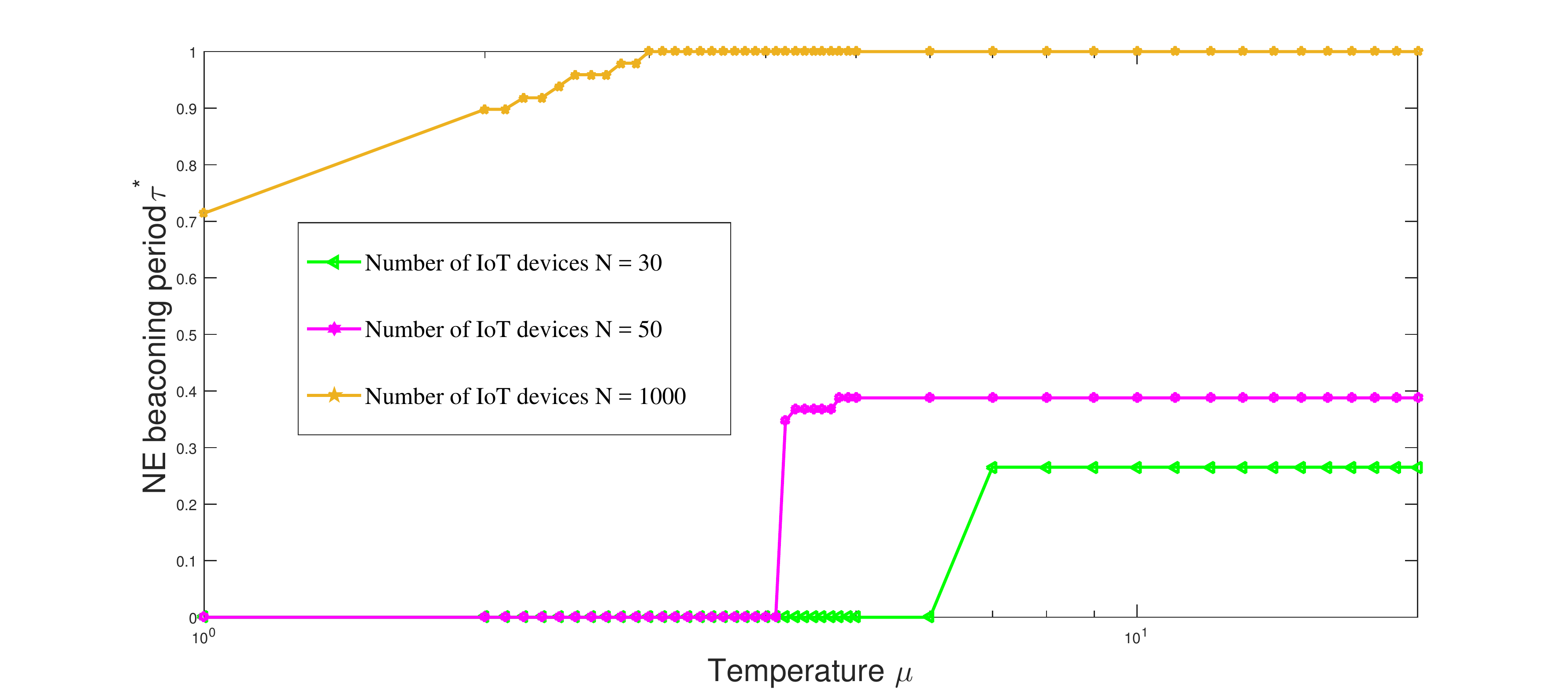}
\caption{NE beaconing periods evolution  for operators with respect to $\mu$ in different populations sizes}}
\end{figure}

Pricing and beaconing learning evolution under different coverage probability values are illustrated in figures 11, respectively 12.
In figure 11, the NE price is minimal with low coverage and from a certain coverage probability value (nearly 0.6) the best strategy tends to $f_{max}$. However, figure 12 shows clearly the linear impact of $P_{cov}$ on NE beaconing strategies. This last are increasing with $P_{cov}$, that confirms results in figure 5.\\

\begin{figure}
\centering{
\includegraphics [width=9.5cm,  height=5cm]{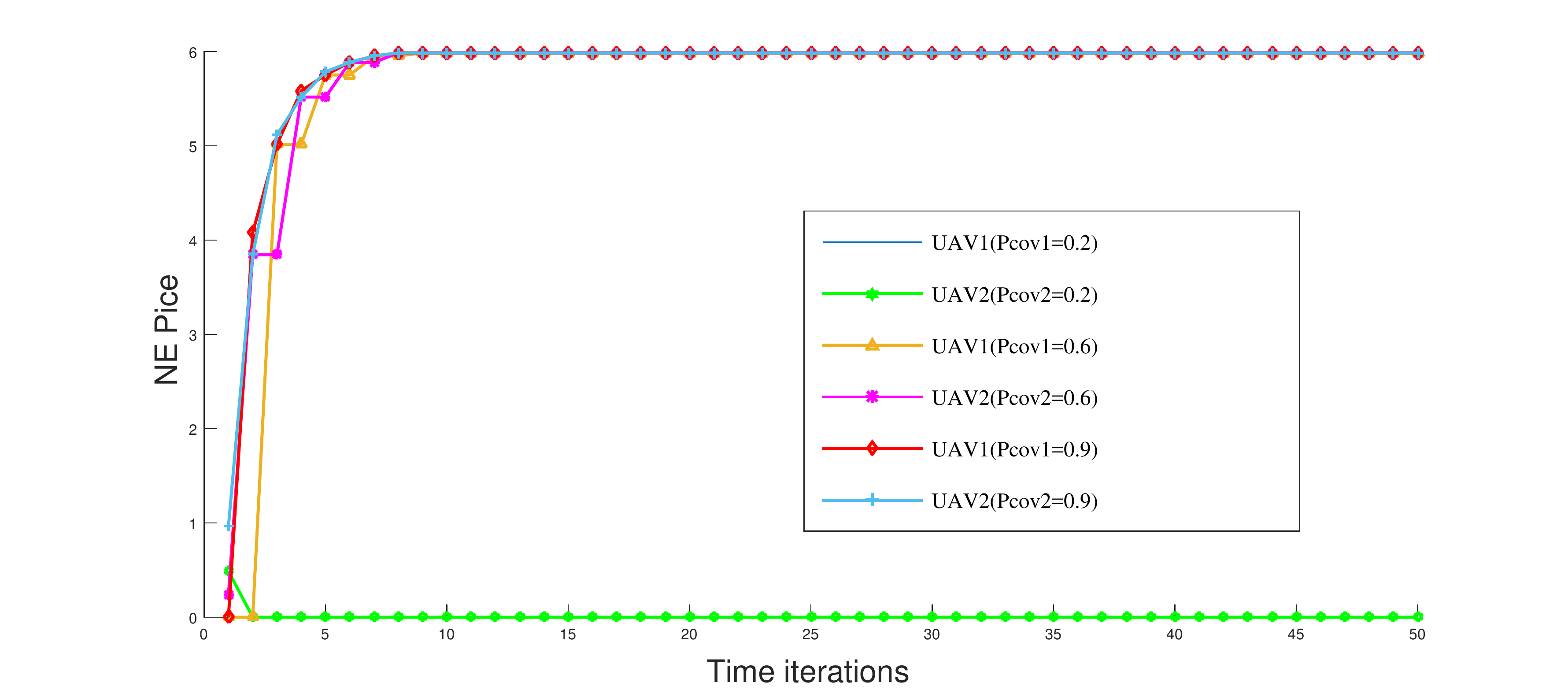}
\caption{NE price with coverage probability under best response dynamic for two homogeneous operators}}
\end{figure}

\begin{figure}
\centering{
\includegraphics [width=9.5cm,  height=5cm]{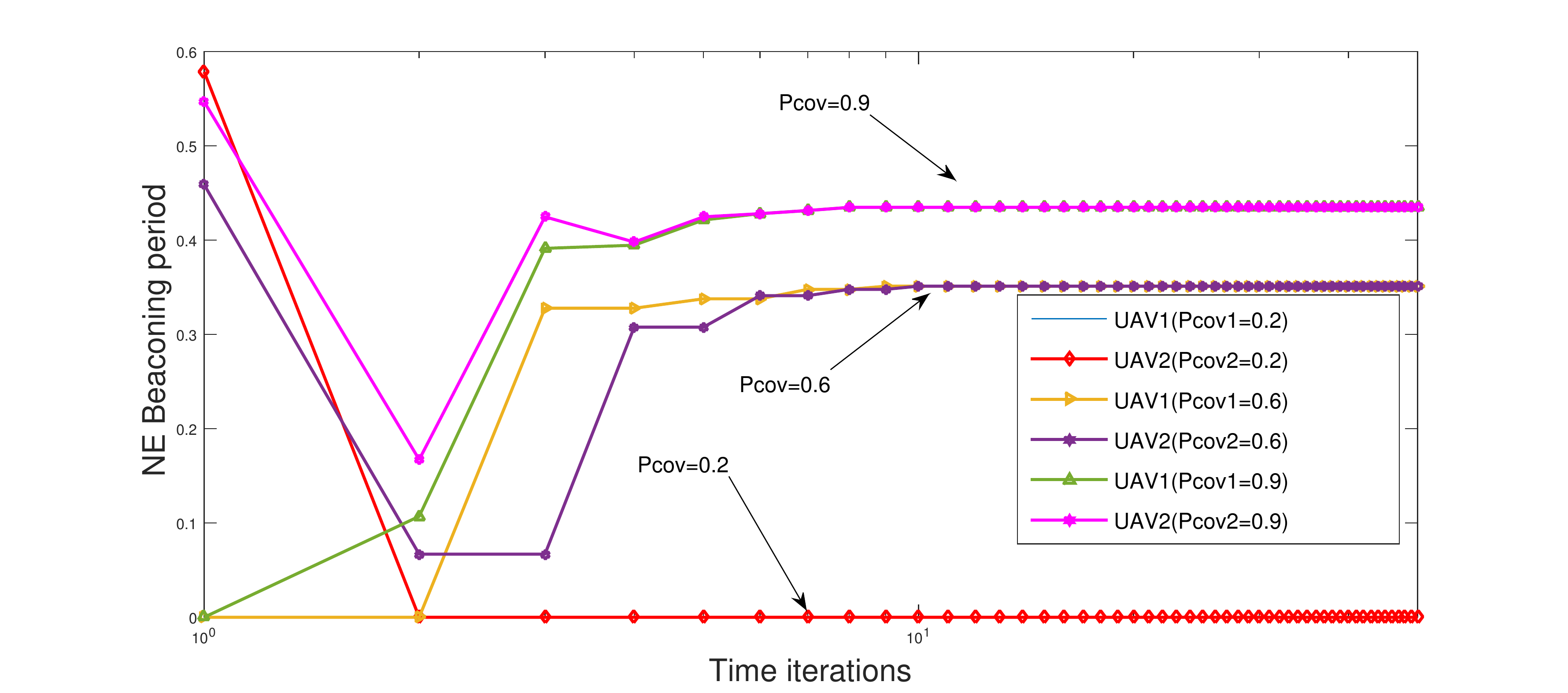}
\caption{NE beaconing period with coverage probability parameter under best response dynamic for two homogeneous operators}}
\end{figure}

Figure 13 respectively figure 14 illustrates the NE price with coverage probability
respectively NE beaconing period for two symmetric UAVs and different IoT population sizes. These two figures are provided  to study the impact of coverage probability on the learning results while considering $f_{min}=2$. Figure 13 shows clearly the "natural" impact of population density on the service price, with a greater population the operators can set a greater prices. Furthermore an interesting finding  in figure 14 is that, with a grater population the NE beaconing duration is maximized, confirming thereby,  the results in figure 13. It is worth mentioning how fast and accurate the proposed algorithm is for the convergence to the unique Nash equilibrium, which provides numerous insights on how to implement it and measure its effictiveness.
\begin{figure}
\centering{
\includegraphics [width=9.5cm,  height=5cm]{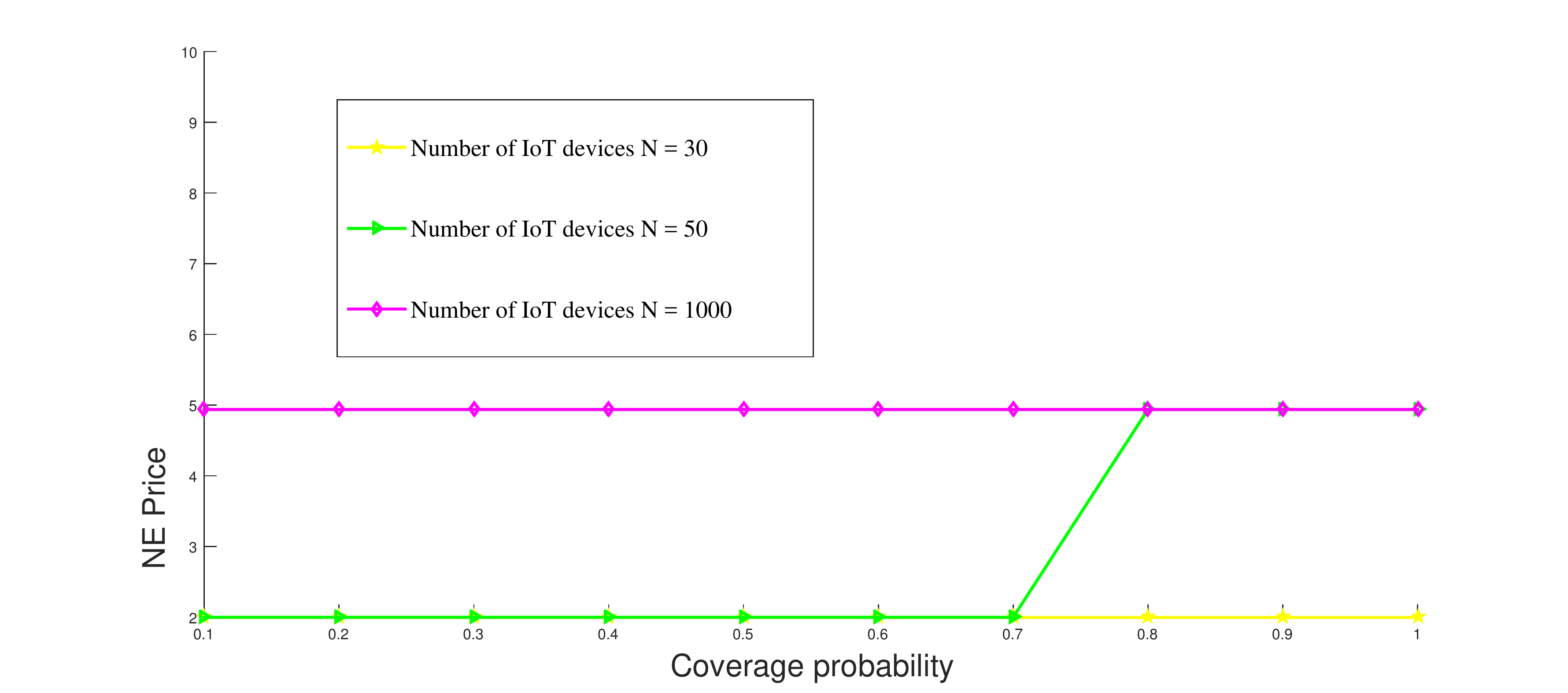}
\caption{NE Price evolution  for operators with coverage probability and different populations sizes}}
\end{figure}
\begin{figure}
\centering{
\includegraphics [width=9.5cm,  height=5cm]{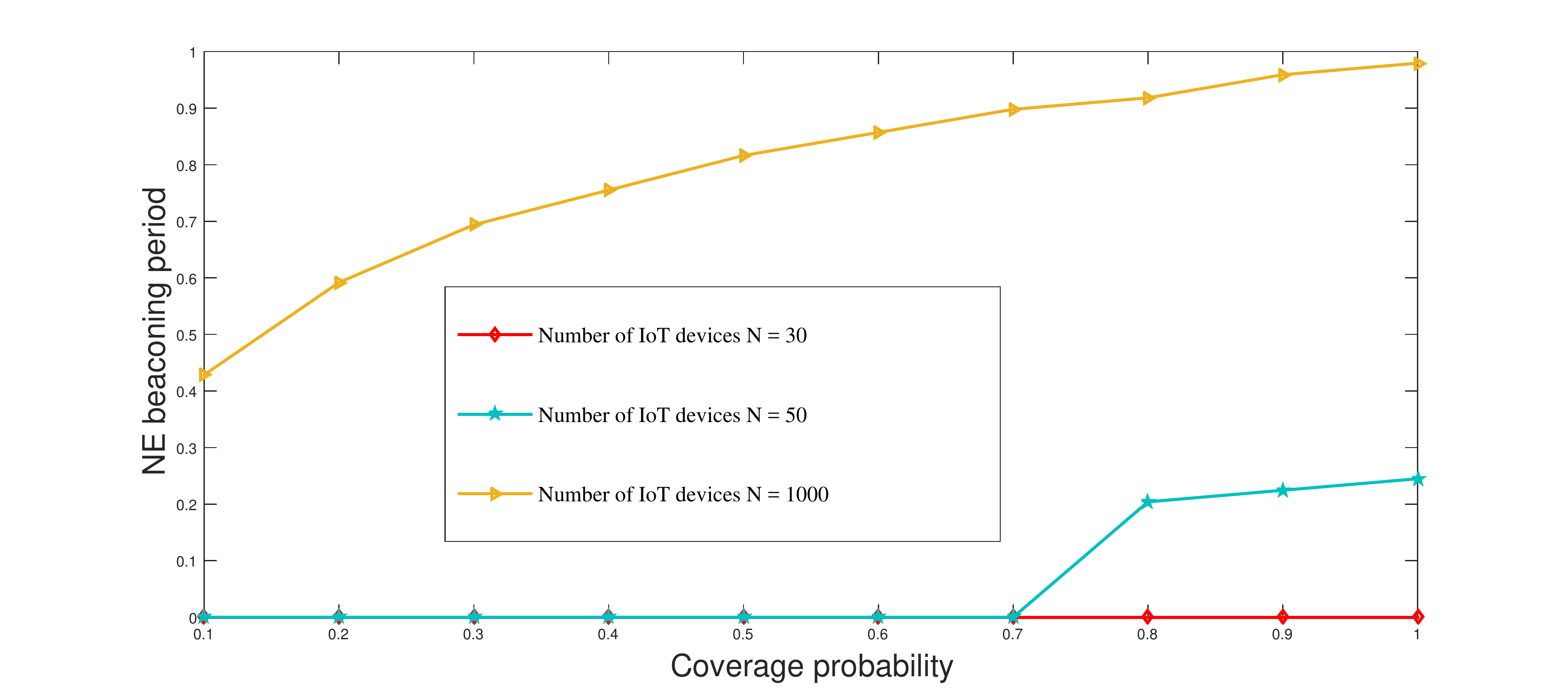}
\caption{NE beaconing periods evolution  for operators with coverage probability and different populations sizes}}
\end{figure}

\section{Conclusion}

In this article, we deal with the Pricing-Availability interaction among adversarial unmanned aerial vehicles acting as flying base stations. we construct a theoretic framework based on non-cooperative game theory and characterize the equilibrium strategies for each UAV, both in terms of equilibrium fees and equilibrium availability probability. Furthermore, a special feature is obtained as the availability game with fixed price is sub-modular while the pricing game is super-modular under fixed availability. Next, we check that a simple iterative best response-based algorithm allows to explore the unique Nash equilibrium of the game. Performance evaluation at equilibrium results allow UAVs service providers to optimize their energy consumption while maximizing their monetary revenues.\\ 

As a future work, we are generalizing our proposal while considering heterogeneous mobility motions. Furthermore, some field experiments are also envisioned.

\section*{Funding}
This work has been conducted within the framework of $Mobicity$ Project funded by The Moroccan Ministry of Higher Education and Scientific Research, and the National Centre for Scientific and Technical Research.

\bibliographystyle{unsrt}

\bibliography{mabiblio}

\end{document}